\documentclass[aps,showpacs,pra,twocolumn]{revtex4-1}

\usepackage{exscale}
\usepackage{graphicx}
\usepackage{amsmath}
\usepackage{latexsym}
\usepackage[caption=false]{subfig}
\usepackage{amsfonts}
\usepackage{amssymb}
\usepackage{bbm}
\usepackage{times}
\usepackage[T1]{fontenc}
\usepackage{lipsum}
\usepackage{amsthm}
\usepackage{fancyhdr,txfonts,bbm}

\usepackage{tikz}
\usetikzlibrary{matrix,calc}
\usetikzlibrary{shapes.misc}

\usepackage{bbold}
\usepackage{bm}
\usepackage{subfig}
\usepackage{color}

\usepackage{epsfig,pstricks}
\usepackage{changes}
\definechangesauthor[name={Markus}, color=red]{M}
\definechangesauthor[name={Dario}, color=magenta]{D}

\newtheorem{theorem}{Theorem}

\newcommand{{\Cd}}{{\mathbb{C}^d}}
\newcommand{{\C}}{{\mathbb{C}}}

\DeclareMathOperator{\Tr}{Tr}

\begin{document}

\title{Nine Lorentz covariant bitensors for two Dirac spinors to indicate all entangled states}

\author{Markus Johansson}
\affiliation{Barcelona (Barcelona), Spain}
\date{\today}
\begin{abstract}
The spinorial degrees of freedom of two spacelike separated Dirac particles are considered and a collection of nine locally Lorentz covariant bitensors is constructed. Four of these bitensors have been previously described in [Phys. Rev. A {\bf 105}, 032402 (2022), arXiv:2103.07784].
The collection of bitensors has the property that all nine bitensors are simultaneously zero if and only if the state of the two particles is a product state.
Thus this collection of bitensors indicates any type of spinor entanglement between two spacelike separated Dirac particles.

\end{abstract}
\maketitle
\section{Introduction}

Quantum entanglement is a physical property of a composite quantum system that cannot be described in terms of local variables \cite{epr,bell,chsh,bell2}. 
A composite system with spacelike separated subsystems is entangled if it is in a superposition where a physical property of one subsystem is conditioned on a physical property of another subsystem \cite{epr,bell,chsh,bell2,schrodinger1,schrodinger2,schrodinger3}.
Entanglement makes possible nonlocal causation between spacelike separated events \cite{epr,bell,chsh,bell2}, sometimes called action at a distance.
Some physical phenomena and applications are impossible without such nonlocal causation. These include the violation of a Bell inequality \cite{bell,chsh}, quantum steering \cite{schrodinger1,wise}, superdense coding \cite{wiesner} and quantum teleportation \cite{bennett}. The description of spin entanglement for non-relativistic spin-$\frac{1}{2}$ particles has been extensively developed (See e.g. Refs. \cite{ghz,grassl,wootters,wootters2,popescu,pop2,ekert,lind2,kempe,toni,car2,higuchi,sud,tarrach,dur,verstraete2,coffman,luque}).

In Special and General Relativity a locally Lorentz covariant function of two different spacetime points is called a two-point-tensor or {\it bitensor} \cite{ruse,ruse2,synge2,synge,dewitt}. Bitensors can be used to describe physical properties that depend on two spacetime points. In particular they can be used to describe physical properties of this kind that cannot be described in terms of local variables.
For any pair of points that a bitensor is defined on it transforms under local Lorentz transformations in both of the points.
Locally it can transform for example as a Lorentz scalar, pseudoscalar, or vector. Similarly to Ref. \cite{dewitt} we here refer to bitensors by their local transformation properties. For example a bitensor that transforms as a Lorentz scalar in both points is here called a scalar-scalar or {\it bi-scalar}, a bitensor that transforms as a Lorentz vector in both points is called a vector-vector or {\it bi-vector}, and a bitensor that transform as a vector in one point and a scalar in the other is called a vector-scalar.

In relativistic quantum mechanics the spinorial degrees of freedom of a spin-$\frac{1}{2}$ particle are described by a four component Dirac spinor \cite{dirac2,dirac}. We refer here to particles described by such Dirac spinors collectively as Dirac particles.
Dirac particles are used to describe fundamental spin-$\frac{1}{2}$ particles such as quarks and leptons in the Standard Model \cite{schwartz}. They are also used to describe composite particles in some models such as the Yukawa model of hadrons where they describe spin-$\frac{1}{2}$ baryons \cite{yukawa}. 
The description of entanglement for Dirac particles has been investigated in multiple works \cite{czachor,alsing,terno,adami,pachos,ahn,terno2,tera,tera2,mano,won,caban3,caban,geng,leon,delgado,moradi,caban2,tessier,spinorent,multispinor,lorent}, but remains less developed than the description for non-relativistic spin-$\frac{1}{2}$ particles. 
One way to describe the entanglement of Dirac particles is by using locally Lorentz covariant bitensors.
The description of the entanglement of spinorial degrees of freedom in a system of two Dirac particles using bitensors that transform locally as Lorentz scalars or pseudoscalars has been considered in Refs. \cite{spinorent,lorent}. 

In this work we consider locally Lorentz covariant bitensors for a system of two spacelike separated Dirac particles. In particular we consider the bitensors that are identically zero for all product states of the two particles. We construct a collection of 9 such bitensors and show that this collection has the property that all the bitensors are simultaneously zero if and only if the state of the two Dirac spinors is a product state. This equivalence relation is the main result of this work. The collection of locally Lorentz covariant bitensors includes 4 bitensors previously introduced in Ref. \cite{spinorent}. These 4 bitensors transform locally either as a Lorentz scalar or pseudoscalar.
The additional 5 bitensors introduced here transform in one or both spacetime points as a Lorentz vector.

The outline is as follows. Sections \ref{dir}-\ref{covariants} present the relevant background material, discuss the physical assumptions made in this work and introduce the tools used to construct the locally Lorentz covariant bitensors.
In particular, section \ref{dir} gives the description of Dirac particles and discusses the physical assumptions. In section \ref{rep} we describe the Lorentz group and its spinor representation. Section \ref{covariants} describes how to construct Lorentz covariants from skew-symmetric bilinear forms. 
Sections \ref{conn} and \ref{rulethemall} contain the results.
In particular, section \ref{conn} describes the nine locally Lorentz covariant bitensors.
Section \ref{rulethemall} gives the statement and proof of the main result.
Section \ref{diss} is the discussion and conclusions.

\section{Dirac particles}\label{dir}
In relativistic quantum mechanics a spin-$\frac{1}{2}$ particle, or Dirac particle, is described by the Dirac equation \cite{dirac2,dirac} given in natural units $\hbar=c=1$ by
\begin{eqnarray}\label{dirr}
\left[\sum_{\mu}\gamma^\mu(i\partial_{\mu}-qA_{\mu}(x)) -m\right]\psi(x)=0,
\end{eqnarray}
where $m$ is the mass, $q$ is the electromagnetic charge, and $A_{\mu}(x)$ is the four-potential. The Dirac spinor $\psi(x)$ in Eq. (\ref{dirr}) is a four component object defined by
\begin{eqnarray}\label{spinor}
\psi(x)\equiv
\begin{pmatrix}
\psi_0(x) \\
\psi_1(x)\\
\psi_2(x) \\
\psi_3(x) \\
\end{pmatrix},
\end{eqnarray}
where each component is a complex valued function of the four-vector $x$. The $\gamma^\mu$ in Eq. (\ref{dirr}) for $\mu=0,1,2,3$ are $4\times 4$ matrices
defined by the relations
\begin{eqnarray}\label{anti}
\gamma^\mu\gamma^\nu+\gamma^\nu\gamma^\mu=2g^{\mu\nu}I,
\end{eqnarray}
where $g^{\mu\nu}$ is the Minkowski metric with signature $(+---)$.
The matrices $\gamma^0$, $\gamma^1$, $\gamma^2$, and $\gamma^3$ are not uniquely defined by the relations in Eq. (\ref{anti}) and have to be chosen by convention.
One such conventional choice is the so called Dirac matrices or {\it gamma matrices} given by
\begin{align}
\gamma^0&=
\begin{pmatrix}
I & 0 \\
0 & -I  \\
\end{pmatrix},
&\gamma^1=
\begin{pmatrix}
0  & \sigma^1 \\
-\sigma^1 &  0 \\
\end{pmatrix},\nonumber\\
\gamma^2&=
\begin{pmatrix}
0  & \sigma^2 \\
-\sigma^2 &  0 \\
\end{pmatrix},
&\gamma^3=
\begin{pmatrix}
0  & \sigma^3 \\
-\sigma^3 &  0 \\
\end{pmatrix},
\end{align}
where $\sigma^1,\sigma^2,\sigma^3$ are the Pauli matrices and $I$ is the $2\times 2$ identity matrix
\begin{eqnarray}
I=
\begin{pmatrix}
1 & 0 \\
0 & 1  \\
\end{pmatrix},\phantom{o}
\sigma^1=
\begin{pmatrix}
0  & 1 \\
1 &  0 \\
\end{pmatrix},\phantom{o}
\sigma^2=
\begin{pmatrix}
0  & -i \\
i &  0 \\
\end{pmatrix},\phantom{o}
\sigma^3=
\begin{pmatrix}
1  & 0 \\
0 &  -1 \\
\end{pmatrix}.
\end{eqnarray}
For a derivation and discussion of the properties of the gamma matrices see e.g. Ref. \cite{dirac2}, Ref. \cite{dirac} Ch. XI and Ref. \cite{pauli}.

We can consider two products of gamma matrices that are used in the following.
One such product is the matrix
\begin{eqnarray}
C\equiv i\gamma^1\gamma^3=
\begin{pmatrix}
-\sigma^2 & 0 \\
0 & -\sigma^2  \\
\end{pmatrix}.
\end{eqnarray}
For every gamma matrix $\gamma^\mu$ and its transpose $\gamma^{\mu T}$ we have that
\begin{eqnarray}\label{uub}
 C\gamma^\mu =\gamma^{\mu T}C.
\end{eqnarray}
The matrix $C$ is anti-symmetric and its own inverse, i.e., $C=-C^T=C^{-1}$.
The other product of gamma matrices is 
\begin{eqnarray}
\gamma^5\equiv i\gamma^0\gamma^1\gamma^2\gamma^3=
\begin{pmatrix}
0 & I \\
I & 0  \\

\end{pmatrix},
\end{eqnarray}
which anticommutes with $\gamma^\mu$ for all $\mu$
\begin{eqnarray}
\gamma^5\gamma^\mu=-\gamma^\mu\gamma^5.
\end{eqnarray}
The matrix $\gamma^5$ is symmetric and its own inverse, i.e., $\gamma^5=(\gamma^{5})^{T}=(\gamma^{5})^{-1}$.

A solution to the Dirac equation can be expressed using a set of four orthogonal basis spinors $\phi_0$, $\phi_1$, $\phi_2$, and $\phi_3$ as
\begin{eqnarray}\label{bass}
\psi(x)=\sum_{j=0,1,2,3}\psi_{j}(x)\phi_j,
\end{eqnarray}
where the coefficients $\psi_{j}(x)$ are functions of the four vector $x$, and the four basis spinors are
\begin{eqnarray}\label{basis}
{\phi_0}=
\begin{pmatrix}
1 \\
0   \\
0   \\
0   \\
\end{pmatrix},\phantom{o}
{\phi_1}=
\begin{pmatrix}
0 \\
1   \\
0   \\
0   \\
\end{pmatrix},\phantom{o}
{\phi_2}=
\begin{pmatrix}
0 \\
0   \\
1   \\
0   \\
\end{pmatrix},\phantom{o}
{\phi_3}=
\begin{pmatrix}
0 \\
0   \\
0   \\
1   \\
\end{pmatrix}.
\end{eqnarray}
In the following we use this basis of spinors to describe the states of Dirac particles.

In this work we consider a scenario with two Dirac particles at spacelike separation.
As in Refs. \cite{spinorent,multispinor,lorent} we introduce laboratories and assume that each laboratory contains only one Dirac particle.
Furthermore, as has been done in References \cite{alsing,pachos,moradi,caban2,caban3,spinorent,multispinor,lorent} we assume that the state of two spacelike separated particles that have not previously interacted can be described as a tensor product  $\psi_A(x_A)\otimes \varphi_B(x_B)$ of single particle states. We also assume that the tensor products of the elements of the single particle spinor bases $\phi_{j_A}\otimes \phi_{k_B}$ is a basis for the two-particle states.

The assumption that a tensor product structure can be used to describe the state of two particles at spacelike separation is commonly made but is not trivial.
The motivation for this assumption is that the operations on one of the particles in such a system can be made jointly with the operations on the other particle, i.e., the spacelike separated operations commute.
However, it is not known if  a description where the operations on the different particles commute is always equivalent to a description where the Hilbert space and the algebra of operations has a tensor product structure \cite{navascues,tsirelson,werner}. This open question is known as Tsirelson's Problem \cite{tsirelson}. Nevertheless, if for each particle the algebra of operations is finite dimensional it has been shown in References \cite{tsirelson,werner} that a description with commuting operations is equivalent to a description with a tensor product structure. In particular this equivalence holds if the Hilbert space of the shared system has finite dimension. In any real world experiment an  operationally constructed Hilbert space with sufficiently large finite dimension can be used to describe the system (See Appendix \ref{opp} for a discussion). Therefore we assume that the use of a tensor product structure is operationally motivated.

As in Refs. \cite{spinorent,multispinor,lorent} we describe each Dirac particle as being in its own Minkowski space. Such a Minkowski space should be understood as the local description of spacetime used by the laboratory holding the particle. 
If the two spacelike separated laboratories are in a flat spacetime the two different Minkowski spaces are the two laboratories different descriptions of the same Minkowski spacetime they are both in. If instead the two spacelike separated laboratories are in a curved spacetime described by General Relativity (See e.g. Ref. \cite{wald}) the two Minkowski spaces are the respective Minkowski tangent-spaces of two spacelike separated spacetime points.

\section{The Lorentz group}\label{rep}

A spacetime in General Relativity is described by a four-dimensional manifold. In general such a spacetime manifold has nonzero curvature.
At every non-singular point of a curved spacetime manifold one can define a four-dimensional tangent vector space. Each such tangent space is isomorphic to the Minkowski space (See e.g. Ref. \cite{wald}).
As in References \cite{spinorent,multispinor,lorent} we assume a scenario where the local curvature of spacetime is sufficiently small so that it is physically motivated to ignore it. We then describe a Dirac particle as being in the Minkowski tangent space of a point instead of being in the spacetime manifold itself.

A local Lorentz transformation in a spacetime point is a coordinate transformation on the Minkowski tangent space to the point. These coordinate transformations include the so called proper orthochronous Lorentz transformations that are combinations of rotations and Lorentz boosts. The proper orthochronous Lorentz transformations preserve both the direction of time and the orientation of the spatial vectors. Additionally the Lorentz transformations also include the time reversal T that changes the direction of time and the parity inversion P that changes the sign of all spatial vectors.

A Lorentz transformation $\Lambda$ on a Minkowski tangent space to a spacetime point induces a transformation on a Dirac spinor in the point given by the spinor representation $S(\Lambda)$ of $\Lambda$. The spinor transforms under $\Lambda$ as $\psi(x)\to \psi'(x')=S(\Lambda)\psi(x)$ where $x'=\Lambda x$ (See e.g. Ref. \cite{zuber}), and the Dirac equation transforms as
\begin{eqnarray}
&&\left[\sum_{\mu}\gamma^\mu(i\partial_{\mu}-qA_{\mu}) -m\right]\psi(x)=0\nonumber\\
\to&&\left[\sum_{\mu,\nu}\gamma^\mu(\Lambda^{-1})^{\nu}_{\mu}(i\partial_{\nu}-qA_{\nu}) -m\right]S(\Lambda)\psi(x)=0.
\end{eqnarray}
The invariance of the Dirac equation implies that the gamma matrices transform as
\begin{eqnarray}\label{vecc}
S^{-1}(\Lambda)\gamma^\mu S(\Lambda)=\sum_\nu\Lambda^{\mu}_{\nu}\gamma^\nu.
\end{eqnarray}
We can see from Eq. (\ref{vecc}) that the gamma matrices $\gamma^\mu$ transform like the components of a four-vector under Lorentz transformations.

The Lorentz group is the group of all Lorentz transformations.
It is a six-dimensional Lie group with four connected components. The connected component of the Lorentz group that contains the identity element, the so called proper orthochronous Lorentz group, is the Lie group of all proper orthochronous Lorentz transformations.
Similarly, the spinor representation of the Lorentz group is the group of all the spinor representations of Lorentz transformations.
It is also a six-dimensional Lie group with four connected components. The connected component of this group that contains the identity element, the so called spinor representation of the proper orthochronous Lorentz group, is the Lie group of all spinor representations of proper orthochronous Lorentz transformations.
This connected Lie group is generated by the exponentials of a Lie algebra. 
The six generators $S^{\rho\sigma}$ of this Lie algebra are defined by
\begin{eqnarray}\label{gene}
S^{\rho\sigma}=\frac{1}{4}[\gamma^\rho,\gamma^\sigma]=\frac{1}{2}\gamma^\rho\gamma^\sigma-\frac{1}{2}g^{\rho\sigma}I,
\end{eqnarray}
where $g^{\rho\sigma}$ is the Minkowski metric with signature $(+---)$.
The exponentials of the three generators $S^{12}, S^{13}$, and $S^{23}$ generate the spinor representations of the spatial rotations while the exponentials of the three generators $S^{01},S^{02}$, and $S^{03}$ generate the spinor representations of the Lorentz boosts.
For any element $\frac{1}{2}\sum_{\rho,\sigma} \omega_{\rho\sigma}S^{\rho\sigma}$ of the Lie algebra, where the $\omega_{\rho\sigma}$ are real numbers, the matrix exponential of the element is a finite transformation
 \begin{eqnarray}
S(\Lambda)=\exp\left(\frac{1}{2}\sum_{\rho,\sigma} \omega_{\rho\sigma}S^{\rho\sigma}\right).
\end{eqnarray}
Any spinor representation of a proper orthochronous Lorentz transformation
can be decomposed as a product of such exponentials. See e.g. Ref. \cite{zuber}.

By composing the parity inversion P, the time reversal T and the combined PT transformation with the elements of the proper orthochronous Lorentz group we obtain the elements of the other three connected components of the Lorentz group. In the same way, by composing the spinor representations of the parity inversion P, the time reversal T, and the PT transformation with the elements of the spinor representation of the proper orthochronous Lorentz group we obtain the elements of the other three connected components of the spinor representation of the Lorentz group. 

The spinor representations of the P and T transformations are only defined up to a multiplicative U(1) factor and must therefore be chosen by convention.
Here we choose the spinor representation of the parity inversion P as
\begin{eqnarray}
S(\textrm{P})=\gamma^0.
\end{eqnarray}
The spinor representation of the time reversal T is up to a multiplicative U(1) factor defined as a multiplication by the matrix $C$ and a complex conjugation of the spinor and we chose it here as $\psi \to C\psi^*$. 
See e.g. Ref. \cite{bjorken} Ch. 5.
We use these choices of the spinor representations of the parity inversion P and time reversal T in the following.

\section{Constructing Lorentz covariants from skew-symmetric bilinear forms}
\label{covariants}

A physical quantity that transforms under a representation of the Lorentz group is called a Lorentz covariant. If the quantity is invariant under the Lorentz group it is called a Lorentz scalar. If on the other hand the quantity is invariant under the proper orthochronous Lorentz group but changes sign under the parity inversion P it is called a Lorentz pseudoscalar.
A covariant with four components $\alpha^\mu$ where the components transform under a proper orthochronous Lorentz transformation $\Lambda$ as $\alpha^\mu\to\sum_\nu\Lambda^{\mu}_{\nu}\alpha^\nu$ is called
a Lorentz vector if the $\alpha^\mu$ for $\mu\neq 0$ change sign under parity inversion P while $\alpha^0$ is invariant. If instead the $\alpha^\mu$ for $\mu\neq 0$ are invariant under parity inversion P while $\alpha^0$ changes sign the four component covariant is called a Lorentz pseudovector.

By considering the definition of the generators $S^{\rho\sigma}$ of the spinor representation of the proper orthochronous Lorentz group given in Eq. (\ref{gene}) and the properties of the matrix $C$ described in Eq. (\ref{uub}) we can see that
\begin{eqnarray}
S^{\rho\sigma T}C=\frac{1}{4}[\gamma^{\sigma T},\gamma^{\rho T}]C=-\frac{1}{4}C[\gamma^{\rho},\gamma^{\sigma}]=-CS^{\rho\sigma}.
\end{eqnarray}
This relation implies that for any finite transformation $S(\Lambda)$ we have that $S(\Lambda)^TC=CS(\Lambda)^{-1}$.
From this we can see that $\psi^TC$ transforms under a proper orthochronous Lorentz transformations as 
\begin{eqnarray}
\psi^TC \to \psi^TCS(\Lambda)^{-1}.
\end{eqnarray}
Using this transformation property a Lorentz scalar can be constructed as a bilinear form $\psi^TC\varphi$ (See e.g. Ref. \cite{pauli}). It transforms under a proper orthochronous Lorentz transformation as
\begin{eqnarray}
\psi^TC\varphi\to\psi^TCS(\Lambda)^{-1}S(\Lambda)\varphi=\psi^TC\varphi. 
\end{eqnarray}
Moreover, $\psi^TC\varphi$ is invariant under the parity inversion P since $\psi^T\gamma^0C\gamma^0\varphi=\psi^TC\varphi$.

A Lorentz pseudoscalar can be constructed as a bilinear form $\psi^TC\gamma^5\varphi$ (See e.g. Ref. \cite{pauli}). It transforms under a proper orthochronous Lorentz transformation as
\begin{eqnarray}
\psi^TC\gamma^5\varphi\to\psi^TCS(\Lambda)^{-1}\gamma^5S(\Lambda)\varphi=\psi^TC\gamma^5\varphi, 
\end{eqnarray}
since $S^{\rho\sigma}\gamma^5=\gamma^5S^{\rho\sigma}$ and therefore $\gamma^5S(\Lambda)=S(\Lambda)\gamma^5$.
Moreover, $\psi^TC\gamma^5\varphi$ changes sign under the parity inversion P since $\psi^T\gamma^0C\gamma^5\gamma^0\varphi=-\psi^TC\gamma^5\varphi$.

Next we can consider a proper orthochronous Lorentz transformation of the bilinear form $\psi^TC\gamma^\mu\varphi$ 
\begin{eqnarray}
\psi^TC\gamma^\mu\varphi\to \psi^TCS(\Lambda)^{-1}\gamma^\mu S(\Lambda)\varphi =\sum_\nu\Lambda^{\mu}_{\nu}\psi^TC\gamma^\nu\varphi,
\end{eqnarray}
where we have used Eq. (\ref{vecc}). Moreover, we have that $\psi^TC\gamma^\mu\varphi$ for $\mu\neq 0$ changes sign under the parity inversion P since $\psi^T\gamma^0C\gamma^\mu\gamma^0\varphi=-\psi^TC\gamma^\mu\varphi$ while $\psi^TC\gamma^0\varphi$ is invariant.
Thus $\psi^TC\gamma^\mu\varphi$ transforms under Lorentz transformations as an element of a Lorentz vector for any $\mu$. Therefore the four components $\psi^TC\gamma^\mu\varphi$ for $\mu=0,1,2,3$ together form a Lorentz vector (See also e.g. Ref. \cite{pauli}).

The scalar $\psi^TC\varphi$ and the pseudoscalar $\psi^TC\gamma^5\varphi$ are both skew-symmetric bilinear forms, i.e., $\psi^TC\varphi=-\varphi^TC\psi$ and $\psi^TC\gamma^5\varphi=-\varphi^TC\gamma^5\psi$ due to the antisymmetry of the matrices $C$ and $C\gamma^5$ respectively. Thus in particular we see that $\psi^TC\psi=0$ and $\psi^TC\gamma^5\psi=0$. Moreover, each component $\psi^TC\gamma^\mu\varphi$ of the vector is also a skew symmetric bilinear form since $(\psi^TC\gamma^\mu\varphi)^T=-\varphi^T\gamma^{\mu T}C\psi=-\varphi^TC\gamma^{\mu}\psi$. Thus in particular we have that $\psi^TC\gamma^\mu\psi=0$ for all $\mu$.

We recall that the U(1) phase factor of the spinor representation of the parity inversion P has been chosen by convention. Therefore the U(1) phase acquired by the scalar $\psi^TC\varphi$ and the phase acquired by pseudoscalar $\psi^TC\gamma^5\varphi$ under the parity inversion P depend on this choice. However, the difference by a factor of $-1$ between the phase acquired by $\psi^TC\gamma^5\varphi$, and the phase acquired by $\psi^TC\varphi$ under the parity inversion P, does not depend of the choice of U(1) phase factor. Likewise the phases acquired by the components $\psi^TC\gamma^\mu\varphi$
of the vector also depend on the choice of U(1) phase factor but the difference by a factor of $-1$ between the phase acquired by $\psi^TC\gamma^\mu\varphi$ for $\mu\neq 0$ and the phase acquired by $\psi^TC\gamma^0\varphi$ does not depend on this choice.

\section{Constructing bitensors that are zero for all product states}\label{conn}

Here we describe a how to construct locally Lorentz covariant bitensors for two spacelike separated Dirac particles. In particular we construct such bitensors that are identically zero for all product states.

Consider two spacelike separated observers each with their own laboratory containing a Dirac particle. We name these two observers Alice and Bob, respectively. Then, we let the two particles be in a joint state and assume that the operations on Alice's particle commute with the operations on Bob's particle. Further, we assume that a tensor product Hilbert space can be used to describe the shared two-particle system and that the tensor products $\phi_{j_A}\otimes \phi_{k_B}$ of local basis spinors can be used as a basis. Let $x_A$ and $x_B$ be the coordinates in Alice's Minkowski space and Bob's Minkowski space, respectively.
Then we can express the state in this tensor product basis as
\begin{eqnarray}\label{bilbo}
\psi_{AB}(x_A,x_B)=\sum_{j_A,k_B}\psi_{j_A,k_B}(x_A,x_B)\phi_{j_A}\otimes \phi_{k_B},
\end{eqnarray}
where the coefficients $\psi_{j_A,k_B}(x_A,x_B)$ are complex valued functions of the coordinates $x_A$ and $x_B$.

Any product state of a system of two Dirac particles, i.e., any state that can be created using only local resources, can be completely factorized as $\psi(x_A)\otimes\varphi(x_B)$ for some $\psi(x_A)$ and $\varphi(x_B)$. Any state not on this form
is by definition entangled.

Next we consider the state in Eq. (\ref{bilbo}) and hide the subscripts $A$ and $B$ on the indices of the spinor basis elements and let $\psi_{jk}(x_A,x_B)\equiv \psi_{j_A,k_B}(x_A,x_B)$.
As was done in Refs. \cite{spinorent,lorent} we can arrange the state coefficients $\psi_{jk}(x_A,x_B)$ as a matrix by making $j$ the row index and $k$ the column index. This $4\times 4$ matrix $\Psi_{AB}(x_A,x_B)$ is given by 
\begin{eqnarray}\label{bofur}
&&\Psi_{AB}(x_A,x_B)\nonumber\\
&&\equiv\sum_{jk}\psi_{jk}(x_A,x_B)\phi_{j}\otimes \phi_{k}^T\nonumber\\&&=
\begin{pmatrix}
\psi_{00}(x_A,x_B) & \psi_{01}(x_A,x_B) & \psi_{02}(x_A,x_B) & \psi_{03}(x_A,x_B)\\
\psi_{10}(x_A,x_B) & \psi_{11}(x_A,x_B) & \psi_{12}(x_A,x_B) & \psi_{13}(x_A,x_B)\\
\psi_{20}(x_A,x_B) & \psi_{21}(x_A,x_B) & \psi_{22}(x_A,x_B) & \psi_{23}(x_A,x_B) \\
\psi_{30}(x_A,x_B) & \psi_{31}(x_A,x_B) & \psi_{32}(x_A,x_B) & \psi_{33}(x_A,x_B)\\
\end{pmatrix}.\nonumber\\
\end{eqnarray}
The matrix $\Psi_{AB}(x_A,x_B)$ is always nonzero and thus has rank between 1 and 4.
Note that if $\Psi_{AB}(x_A,x_B)$ is a product state it can be written as a tensor product $\psi(x_A)\otimes \varphi^T(x_B)$ for some spinors $\psi(x_A)$ and $\varphi(x_B)$. In particular $\Psi_{AB}(x_A,x_B)$ is a product state if and only if it has rank one.

The spinor representation $S(\Lambda_A)$ of a proper orthochronous Lorentz transformation $\Lambda_A$ on Alice's particle acts on $\Psi_{AB}(x_A,x_B)$ from the left and the spinor representation $S(\Lambda_{B})$ of a proper orthochronous Lorentz transformation $\Lambda_B$ on Bob's particle acts in transposed form $S(\Lambda_{B})^T$ from the right
\begin{eqnarray}
\Psi_{AB}(x_A,x_B)\to S(\Lambda_A)\Psi_{AB}(x_A,x_B)S(\Lambda_{B})^T.
\end{eqnarray}

Using this transformation property under local proper orthochronous Lorentz transformations we can construct locally Lorentz covariant bitensors that are identically zero for all product states from $\Psi_{AB}(x_A,x_B)$ and the matrices $C$, $C\gamma^5$ and $C\gamma^\mu$ for $\mu=0,1,2,3$.

\subsection{A bi-scalar that is zero for all product states}
We can consider the bi-scalar 
\begin{eqnarray}
I_1=\frac{1}{2}\Tr[\Psi_{AB}^TC\Psi_{AB} C],
\end{eqnarray}
that was introduced in Ref. \cite{spinorent}.
It transforms under a proper orthochronous Lorentz transformation $\Lambda_A$ in Alice's lab as
\begin{eqnarray}
\Tr[\Psi_{AB}^TC\Psi_{AB} C]\to &&\Tr[\Psi_{AB}^TCS(\Lambda_A)^{-1}S(\Lambda_A)\Psi_{AB} C]\nonumber\\
&&=\Tr[\Psi_{AB}^TC\Psi_{AB} C].
\end{eqnarray}
Similarly, it transforms under a proper orthochronous Lorentz transformation $\Lambda_B$ in Bob's lab as
\begin{eqnarray}
\Tr[\Psi_{AB}^TC\Psi_{AB} C]\to &&\Tr[\Psi_{AB}^TC\Psi_{AB} CS(\Lambda_B)^{-1}S(\Lambda_B)]\nonumber\\
&&=\Tr[\Psi_{AB}^TC\Psi_{AB} C].
\end{eqnarray}
Furthermore it is invariant under parity inversion P in both labs since $\gamma^0C\gamma^0=C$.
The bi-scalar $I_{1}$ is identically zero for any product state $\psi\otimes \varphi^T$ since $\psi^TC\psi \varphi^TC\varphi$ is zero for all $\psi,\varphi$ due to the antisymmetry of $C$.
\subsection{A bi-pseudoscalar that is zero for all product states}
The bi-pseudoscalar
\begin{eqnarray}
I_2=\frac{1}{2}\Tr[\Psi_{AB}^TC\gamma^5\Psi_{AB} C\gamma^5],
\end{eqnarray}
was introduced in Ref. \cite{spinorent}.
It transforms under a proper orthochronous Lorentz transformation $\Lambda_A$ in Alice's lab as
\begin{eqnarray}
\Tr[\Psi_{AB}^TC\gamma^5\Psi_{AB} C\gamma^5]\to &&\Tr[\Psi_{AB}^TC\gamma^5 S(\Lambda_A)^{-1}S(\Lambda_A)\Psi_{AB} C\gamma^5]\nonumber\\
&&=\Tr[\Psi_{AB}^TC\gamma^5\Psi_{AB} C\gamma^5].
\end{eqnarray}
Similarly, it transforms under a proper orthochronous Lorentz transformation $\Lambda_B$ in Bob's lab as
\begin{eqnarray}
\Tr[\Psi_{AB}^TC\gamma^5\Psi_{AB} C\gamma^5]\to &&\Tr[\Psi_{AB}^TC\gamma^5\Psi_{AB} C\gamma^5 S(\Lambda_B)^{-1}S(\Lambda_B)]\nonumber\\
&&=\Tr[\Psi_{AB}^TC\gamma^5\Psi_{AB} C\gamma^5].
\end{eqnarray}
Furthermore it changes sign under a parity inversion P in Alice's lab and under a parity inversion P in Bob's lab since $\gamma^0C\gamma^5\gamma^0=-C\gamma^5$.
The bi-pseudoscalar $I_{2}$ is identically zero for any product state $\psi\otimes \varphi^T$ since $\psi^TC\gamma^5\psi \varphi^TC\gamma^5\varphi$ is zero for all $\psi,\varphi$ due to the antisymmetry of $C\gamma^5$.

\subsection{A scalar-pseudoscalar that is zero for all product states}
We can consider the scalar-pseudoscalar
\begin{eqnarray}
I_{2A}&&=\frac{1}{2}\Tr[\Psi_{AB}^TC\Psi_{AB} C\gamma^5],
\end{eqnarray}
that was introduced in Ref. \cite{spinorent}.
It transforms under a proper orthochronous Lorentz transformation $\Lambda_A$ in Alice's lab as
\begin{eqnarray}
\Tr[\Psi_{AB}^TC\Psi_{AB} C\gamma^5]\to &&\Tr[\Psi_{AB}^TC S(\Lambda_A)^{-1}S(\Lambda_A)\Psi_{AB} C\gamma^5]\nonumber\\
&&=\Tr[\Psi_{AB}^TC\Psi_{AB} C\gamma^5].
\end{eqnarray}
Similarly, it transforms under a proper orthochronous Lorentz transformation $\Lambda_B$ in Bob's lab as
\begin{eqnarray}
\Tr[\Psi_{AB}^TC\Psi_{AB} C\gamma^5]\to &&\Tr[\Psi_{AB}^TC\Psi_{AB} C\gamma^5 S(\Lambda_B)^{-1}S(\Lambda_B)]\nonumber\\
&&=\Tr[\Psi_{AB}^TC\Psi_{AB} C\gamma^5].
\end{eqnarray}
Furthermore it changes sign under a parity inversion P in Bob's lab since $\gamma^0C\gamma^5\gamma^0=-C\gamma^5$ but not in Alice's lab since $\gamma^0C\gamma^0=C$.
The scalar-pseudoscalar $I_{2A}$ is identically zero for any product state $\psi\otimes \varphi^T$ since $\psi^TC\psi \varphi^TC\gamma^5\varphi$ is zero for all $\psi,\varphi$ due to the antisymmetry of $C\gamma^5$ and $C$.

\subsection{A pseudoscalar-scalar that is zero for all product states}
The pseudoscalar-scalar
\begin{eqnarray}
I_{2B}&&=\frac{1}{2}\Tr[\Psi_{AB}^TC\gamma^5\Psi_{AB} C],
\end{eqnarray}
was introduced in Ref. \cite{spinorent}.
It transforms under a proper orthochronous Lorentz transformation $\Lambda_A$ in Alice's lab as
\begin{eqnarray}
\Tr[\Psi_{AB}^TC\gamma^5\Psi_{AB} C]\to &&\Tr[\Psi_{AB}^TC\gamma^5 S(\Lambda_A)^{-1}S(\Lambda_A)\Psi_{AB} C]\nonumber\\
&&=\Tr[\Psi_{AB}^TC\gamma^5\Psi_{AB} C].
\end{eqnarray}
Similarly, it transforms under a proper orthochronous Lorentz transformation $\Lambda_B$ in Bob's lab as
\begin{eqnarray}
\Tr[\Psi_{AB}^TC\gamma^5\Psi_{AB} C]\to &&\Tr[\Psi_{AB}^TC\gamma^5\Psi_{AB} C S(\Lambda_B)^{-1}S(\Lambda_B)]\nonumber\\
&&=\Tr[\Psi_{AB}^TC\gamma^5\Psi_{AB} C].
\end{eqnarray}
Furthermore it changes sign under a parity inversion P in Alice's lab since $\gamma^0C\gamma^5\gamma^0=-C\gamma^5$ but not in Bob's lab since $\gamma^0C\gamma^0=C$.
The pseuoscalar-scalar $I_{2B}$ is identically zero for any product state $\psi\otimes \varphi^T$ since $\psi^TC\gamma^5\psi \varphi^TC\varphi$ is zero for all $\psi,\varphi$ due to the antisymmetry of $C\gamma^5$ and $C$.

\subsection{A vector-scalar that is zero for all product states}
We can construct a vector-scalar $K_{A}$ with components
\begin{eqnarray}
K^{\mu}_{A}=\frac{1}{2}\Tr[\Psi_{AB}^TC\gamma^\mu\Psi_{AB} C].
\end{eqnarray}
It transforms as a vector under Lorentz transformations in Alice's lab since
for a proper orthochronous Lorentz transformation $\Lambda_A$ we have that
\begin{eqnarray}
\Tr[\Psi_{AB}^TC\gamma^\mu\Psi_{AB} C]\to && \Tr[\Psi_{AB}^TCS(\Lambda_A)^{-1}\gamma^\mu S(\Lambda_A)\Psi_{AB} C]\nonumber\\
&&=\sum_\nu\Lambda^{\mu}_{A\nu }\Tr[\Psi_{AB}^TC\gamma^\nu\Psi_{AB} C],
\end{eqnarray}
and for a parity inversion P we have that $\gamma^0C\gamma^\mu\gamma^0=-C\gamma^\mu$ for $\mu\neq 0$ and $\gamma^0C\gamma^0\gamma^0=C\gamma^0$.
It transforms as a scalar under Lorentz transformations in Bob's lab since
for a proper orthochronous Lorentz transformation $\Lambda_B$ we have that
\begin{eqnarray}
\Tr[\Psi_{AB}^TC\gamma^\mu\Psi_{AB} C]\to && \Tr[\Psi_{AB}^TC \gamma^\mu \Psi_{AB} CS(\Lambda_B)^{-1}S(\Lambda_B)]\nonumber\\
&&=\Tr[\Psi_{AB}^TC\gamma^\mu\Psi_{AB} C],
\end{eqnarray}
and for a parity inversion P we have that $\gamma^0C\gamma^0=C$.
The vector-scalar $K_{A}$ is identically zero for any product state $\psi\otimes \varphi^T$ since $\psi^TC\gamma^\mu\psi \varphi^TC\varphi$ is zero for all $\mu$ and all $\psi,\varphi$ due to the antisymmetry of $C$ and $C\gamma^\mu$ for all $\mu$.
\subsection{A scalar-vector that is zero for all product states}
We can construct a scalar-vector $K_{B}$ with components
\begin{eqnarray}
K^{\mu}_{B}=\frac{1}{2}\Tr[\Psi_{AB}^TC\Psi_{AB} C\gamma^\mu].
\end{eqnarray}
It transforms as a scalar under Lorentz transformations in Alice's lab since
for a proper orthochronous Lorentz transformation $\Lambda_A$ we have that
\begin{eqnarray}
\Tr[\Psi_{AB}^TC\Psi_{AB} C\gamma^\mu]\to && \Tr[\Psi_{AB}^TCS(\Lambda_A)^{-1}S(\Lambda_A)  \Psi_{AB} C\gamma^\mu]\nonumber\\
&&=\Tr[\Psi_{AB}^TC\Psi_{AB} C\gamma^\mu],
\end{eqnarray}
and for a parity inversion P we have that $\gamma^0C\gamma^0=C$.
It transforms as a vector under Lorentz transformations in Bob's lab since
for a proper orthochronous Lorentz transformation $\Lambda_B$ we have that
\begin{eqnarray}
\Tr[\Psi_{AB}^TC\Psi_{AB} C\gamma^\mu]\to && \Tr[\Psi_{AB}^TC \Psi_{AB} CS(\Lambda_B)^{-1}\gamma^\mu S(\Lambda_B)]\nonumber\\
&&=\sum_\nu\Lambda^{\mu}_{B\nu }\Tr[\Psi_{AB}^TC\Psi_{AB} C\gamma^\nu],
\end{eqnarray}
and for a parity inversion P we have that $\gamma^0C\gamma^\mu\gamma^0=-C\gamma^\mu$ for $\mu\neq 0$ and $\gamma^0C\gamma^0\gamma^0=C\gamma^0$.
The scalar-vector $K_{B}$ is identically zero for any product state $\psi\otimes \varphi^T$ since $\psi^TC\psi \varphi^TC\gamma^\mu\varphi$ is zero for all $\mu$ and all $\psi,\varphi$ due to the antisymmetry of $C$ and $C\gamma^\mu$ for all $\mu$.

\subsection{A vector-pseudoscalar that is zero for all product states}

We can construct a vector-pseudoscalar $L_{A}$ with components
\begin{eqnarray}
L^{\mu}_{A}=\frac{1}{2}\Tr[\Psi_{AB}^TC\gamma^\mu\Psi_{AB} C\gamma^5].
\end{eqnarray}
It transforms as a vector under Lorentz transformations in Alice's lab since
for a proper orthochronous Lorentz transformation $\Lambda_A$ we have that
\begin{eqnarray}
\Tr[\Psi_{AB}^TC\gamma^\mu\Psi_{AB} C\gamma^5]\to && \Tr[\Psi_{AB}^TCS(\Lambda_A)^{-1}\gamma^\mu S(\Lambda_A)\Psi_{AB} C\gamma^5]\nonumber\\
&&=\sum_\nu\Lambda^{\mu}_{A\nu }\Tr[\Psi_{AB}^TC\gamma^\nu\Psi_{AB} C\gamma^5],
\end{eqnarray}
and for a parity inversion P we have that $\gamma^0C\gamma^\mu\gamma^0=-C\gamma^\mu$ for $\mu\neq 0$ and $\gamma^0C\gamma^0\gamma^0=C\gamma^0$.
It transforms as a pseudoscalar under Lorentz transformations in Bob's lab since
for a proper orthochronous Lorentz transformation $\Lambda_B$ we have that
\begin{eqnarray}
\Tr[\Psi_{AB}^TC\gamma^\mu\Psi_{AB} C\gamma^5]\to && \Tr[\Psi_{AB}^TC \gamma^\mu \Psi_{AB} C\gamma^5 S(\Lambda_B)^{-1}S(\Lambda_B)]\nonumber\\
&&=\Tr[\Psi_{AB}^TC\gamma^\mu\Psi_{AB} C\gamma^5],
\end{eqnarray}
and for a parity inversion P we have that $\gamma^0C\gamma^5\gamma^0=-C\gamma^5$.
The vector-pseudoscalar $L_{A}$ is identically zero for any product state $\psi\otimes \varphi^T$ since $\psi^TC\gamma^\mu\psi \varphi^TC\gamma^5\varphi$ is zero for all $\mu$ and all $\psi,\varphi$ due to the antisymmetry of $C\gamma^5$ and $C\gamma^\mu$ for all $\mu$.

\subsection{A pseudoscalar-vector that is zero for all product states}

We can construct a pseudoscalar-vector $L_{B}$ with components
\begin{eqnarray}
L^{\mu}_{B}=\frac{1}{2}\Tr[\Psi_{AB}^TC\gamma^5\Psi_{AB} C\gamma^\mu].
\end{eqnarray}
It transforms as a pseudoscalar under Lorentz transformations in Alice's lab since
for a proper orthochronous Lorentz transformation $\Lambda_A$ we have that
\begin{eqnarray}
\Tr[\Psi_{AB}^TC\gamma^5\Psi_{AB} C\gamma^\mu]\to && \Tr[\Psi_{AB}^TC\gamma^5 S(\Lambda_A)^{-1}S(\Lambda_A)  \Psi_{AB} C\gamma^\mu]\nonumber\\
&&=\Tr[\Psi_{AB}^TC\gamma^5\Psi_{AB} C\gamma^\mu],
\end{eqnarray}
and for a parity inversion P we have that $\gamma^0C\gamma^5\gamma^0=-C\gamma^5$.
It transforms as a vector under Lorentz transformations in Bob's lab since
for a proper orthochronous Lorentz transformation $\Lambda_B$ we have that
\begin{eqnarray}
\Tr[\Psi_{AB}^TC\gamma^5\Psi_{AB} C\gamma^\mu]\to && \Tr[\Psi_{AB}^TC\gamma^5 \Psi_{AB} CS(\Lambda_B)^{-1}\gamma^\mu S(\Lambda_B)]\nonumber\\
&&=\sum_\nu\Lambda^{\mu}_{B\nu }\Tr[\Psi_{AB}^TC\gamma^5\Psi_{AB} C\gamma^\nu],
\end{eqnarray}
and for a parity inversion P we have that $\gamma^0C\gamma^\mu\gamma^0=-C\gamma^\mu$ for $\mu\neq 0$ and $\gamma^0C\gamma^0\gamma^0=C\gamma^0$.
The pseudoscalar-vector $L_{B}$ is identically zero for any product state $\psi\otimes \varphi^T$ since $\psi^TC\gamma^5\psi \varphi^TC\gamma^\mu\varphi$ is zero for all $\mu$ and all $\psi,\varphi$ due to the antisymmetry of  $C\gamma^5$ and $C\gamma^\mu$ for all $\mu$.

\subsection{A bi-vector that is zero for all product states}

We can construct a bi-vector $K_{AB}$ with components
\begin{eqnarray}
K^{\mu\nu}_{AB}=\frac{1}{2}\Tr[\Psi_{AB}^TC\gamma^\mu\Psi_{AB} C\gamma^\nu].
\end{eqnarray}
It transforms as a vector under Lorentz transformations in Alice's lab since
for a proper orthochronous Lorentz transformation $\Lambda_A$ we have that
\begin{eqnarray}
\Tr[\Psi_{AB}^TC\gamma^\mu\Psi_{AB} C\gamma^\nu]\to && \Tr[\Psi_{AB}^TCS(\Lambda_A)^{-1}\gamma^\mu S(\Lambda_A)\Psi_{AB} C\gamma^\nu]\nonumber\\
&&=\sum_\rho\Lambda^{\mu}_{A\rho }\Tr[\Psi_{AB}^TC\gamma^\rho\Psi_{AB} C\gamma^\nu],
\end{eqnarray}
and for a parity inversion P we have that $\gamma^0C\gamma^\mu\gamma^0=-C\gamma^\mu$ for $\mu\neq 0$ and $\gamma^0C\gamma^0\gamma^0=C\gamma^0$.
It transforms as a vector under Lorentz transformations in Bob's lab since
for a proper orthochronous Lorentz transformation $\Lambda_B$ we have that
\begin{eqnarray}
\Tr[\Psi_{AB}^TC\gamma^\mu\Psi_{AB} C\gamma^\nu]\to && \Tr[\Psi_{AB}^TC\gamma^\mu \Psi_{AB} CS(\Lambda_B)^{-1}\gamma^\nu S(\Lambda_B)]\nonumber\\
&&=\sum_\rho\Lambda^{\nu}_{B\rho }\Tr[\Psi_{AB}^TC\gamma^\mu\Psi_{AB} C\gamma^\rho],
\end{eqnarray}
and for a parity inversion P we have that $\gamma^0C\gamma^\nu\gamma^0=-C\gamma^\nu$ for $\nu\neq 0$ and $\gamma^0C\gamma^0\gamma^0=C\gamma^0$.
The bi-vector $K_{AB}$ is identically zero for any product state $\psi\otimes \varphi^T$ since $\psi^TC\gamma^\mu\psi \varphi^TC\gamma^\nu\varphi$ is zero for all $\mu,\nu$ and all $\psi,\varphi$ due to the antisymmetry of $C\gamma^\mu$ for all $\mu$.

\section{Nine bitensors to indicate all entanglement}\label{rulethemall}

We are now ready to state the main result of this work.

\begin{theorem}
The nine locally Lorentz covariant bitensors $I_1$, $I_2$, $I_{2A}$, $I_{2B}$, $K_A$, $K_B$, $L_A$, $L_B$, and $K_{AB}$ are simultaneously zero if and only if $\Psi_{AB}$ is a product state.

\end{theorem}
\begin{proof}

It has already been show in Sec. \ref{conn} that the nine bitensors $I_1$, $I_2$, $I_{2A}$, $I_{2B}$, $K_A$, $K_B$, $L_A$, $L_B$, and $K_{AB}$ are zero if the state is a product state, but we repeat the argument here. The matrices $C$, $C\gamma^5$ and $C\gamma^\mu$ for $\mu=0,1,2,3$ are all antisymmetric. When $\Psi_{AB}(x_A,x_B)$ is a product state it can be written as $\psi(x_A)\otimes \varphi(x_B)^T$. Then all the bitensors $I_1$, $I_2$, $I_{2A}$, and $I_{2B}$, and all the components of the bitensors $K_A$, $K_B$, $L_A$, $L_B$, and $K_{AB}$ reduce to products of the bilinear forms $\psi^TC\psi$, $\psi^TC\gamma^5\psi$, $\psi^TC\gamma^\mu \psi$ for $\mu=0,1,2,3$ and $\varphi^TC\varphi$, $\varphi^TC\gamma^5\varphi$, $\varphi^TC\gamma^\mu \varphi$ for $\mu=0,1,2,3$. For any $\psi,\varphi$ these bilinear forms are all zero due to the antisymmetry of $C$, $C\gamma^5$ and $C\gamma^\mu$ for $\mu=0,1,2,3$.

It remains to show that if the nine bitensors $I_1$, $I_2$, $I_{2A}$, $I_{2B}$, $K_A$, $K_B$, $L_A$, $L_B$, and $K_{AB}$ are simultaneously zero the state is a product state.
To do this we note that if all possible $2\times 2$ submatrices of $\Psi_{AB}$ have determinant zero it follows that the rank of $\Psi_{AB}$ is at most 1. We therefore consider the determinants of the 36 different $2\times 2$ submatrices of $\Psi_{AB}$ and relate them to $I_1$, $I_2$, $I_{2A}$, $I_{2B}$, $K_A$, $K_B$, $L_A$, $L_B$, and $K_{AB}$.

To begin we consider the six linear combinations
\begin{eqnarray}\label{trex}
\frac{1}{4} (K^{1}_{B} - i K^{2}_{B} +  K^{01}_{AB} -i K^{02}_{AB})&=&\psi_{00}\psi_{12}-\psi_{02}\psi_{10},\nonumber\\
\frac{1}{4} (K^{1}_{B} + i K^{2}_{B} +  K^{01}_{AB} +i K^{02}_{AB})&=&\psi_{03}\psi_{11}-\psi_{01}\psi_{13},\nonumber\\
\frac{1}{4} (I_1 + K^{0}_{A} + K^{0}_{B} + K^{00}_{AB})&=&\psi_{00}\psi_{11}-\psi_{01}\psi_{10},\nonumber\\
\frac{1}{4} (I_1 + K^{0}_{A} - K^{0}_{B} - K^{00}_{AB})&=&\psi_{02}\psi_{13}-\psi_{03}\psi_{12},\nonumber\\
\frac{1}{4} (I_{2A} + L^{0}_{A} + K^{3}_{B} + K^{03}_{AB})&=&\psi_{02}\psi_{11}-\psi_{12}\psi_{01},\nonumber\\
\frac{1}{4} (I_{2A} + L^{0}_{A} - K^{3}_{B} - K^{03}_{AB})&=&\psi_{00}\psi_{13}-\psi_{10}\psi_{03}.
\end{eqnarray}
The six linear combinations in Eq. (\ref{trex}) are the determinants of the six $2\times 2$ submatrices that involve only elements of the first and second row of $\Psi_{AB}$. 

Next, we consider the six linear combinations
\begin{eqnarray}\label{trex2}
\frac{1}{4} (K^{11}_{AB} - i K^{12}_{AB} - i K^{21}_{AB} - K^{22}_{AB})&=&\psi_{00}\psi_{22}-\psi_{02}\psi_{20},\nonumber\\
\frac{1}{4} (K^{11}_{AB} + i K^{12}_{AB} - i K^{21}_{AB} + K^{22}_{AB})&=&\psi_{03}\psi_{21}-\psi_{01}\psi_{23},\nonumber\\
\frac{1}{4} (K^{1}_{A} - i K^{2}_{A} -  K^{10}_{AB} +i K^{20}_{AB})&=&\psi_{02}\psi_{23}-\psi_{03}\psi_{22},\nonumber\\
\frac{1}{4} (K^{1}_{A} - i K^{2}_{A} +  K^{10}_{AB} -i K^{20}_{AB})&=&\psi_{00}\psi_{21}-\psi_{01}\psi_{20},\nonumber\\
\frac{1}{4} (L^{1}_{A} - i L^{2}_{A} -  K^{13}_{AB} +i K^{23}_{AB})&=&\psi_{00}\psi_{23}-\psi_{03}\psi_{20},\nonumber\\
\frac{1}{4} (L^{1}_{A} - i L^{2}_{A} +  K^{13}_{AB} -i K^{23}_{AB})&=&\psi_{02}\psi_{21}-\psi_{01}\psi_{22}.
\end{eqnarray}
The six linear combinations in Eq. (\ref{trex2}) are the determinants of the six $2\times 2$ submatrices that involve only elements of the first and third row of $\Psi_{AB}$.

Next, we consider the six linear combinations
\begin{eqnarray}\label{trex3}
\frac{1}{4} (L^{1}_{B} - i L^{2}_{B} -  K^{31}_{AB} +i K^{32}_{AB})&=&\psi_{00}\psi_{32}-\psi_{02}\psi_{30},\nonumber\\
\frac{1}{4} (L^{1}_{B} + i L^{2}_{B} -  K^{31}_{AB} -i K^{32}_{AB})&=&\psi_{03}\psi_{31}-\psi_{01}\psi_{33},\nonumber\\
\frac{1}{4} (I_2 - L^{3}_{A} - L^{3}_{B} + K^{33}_{AB})&=&\psi_{00}\psi_{33}-\psi_{03}\psi_{30},\nonumber\\
\frac{1}{4} (I_2 - L^{3}_{A} + L^{3}_{B} - K^{33}_{AB})&=&\psi_{02}\psi_{31}-\psi_{01}\psi_{32},\nonumber\\
\frac{1}{4} (I_{2B} + L^{0}_{B} - K^{3}_{A} - K^{30}_{AB})&=&\psi_{00}\psi_{31}-\psi_{01}\psi_{30},\nonumber\\
\frac{1}{4} (I_{2B} - L^{0}_{B} - K^{3}_{A} + K^{30}_{AB})&=&\psi_{02}\psi_{33}-\psi_{03}\psi_{32}.
\end{eqnarray}
The six linear combinations in Eq. (\ref{trex3}) are the determinants of the six $2\times 2$ submatrices that involve only elements of the first and fourth row of $\Psi_{AB}$.

Next, we consider the six linear combinations
\begin{eqnarray}\label{trex4}
\frac{1}{4} (L^{1}_{B} - i L^{2}_{B} +  K^{31}_{AB} -i K^{32}_{AB})&=&\psi_{20}\psi_{12}-\psi_{22}\psi_{10},\nonumber\\
\frac{1}{4} (L^{1}_{B} + i L^{2}_{B} +  K^{31}_{AB} +i K^{32}_{AB})&=&\psi_{23}\psi_{11}-\psi_{21}\psi_{13},\nonumber\\
\frac{1}{4} (I_2 + L^{3}_{A} + L^{3}_{B} + K^{33}_{AB})&=&\psi_{22}\psi_{11}-\psi_{21}\psi_{12},\nonumber\\
\frac{1}{4} (I_2 + L^{3}_{A} - L^{3}_{B} - K^{33}_{AB})&=&\psi_{20}\psi_{13}-\psi_{23}\psi_{10},\nonumber\\
\frac{1}{4} (I_{2B} + L^{0}_{B} + K^{3}_{A} + K^{30}_{AB})&=&\psi_{20}\psi_{11}-\psi_{21}\psi_{10},\nonumber\\
\frac{1}{4} (I_{2B} - L^{0}_{B} + K^{3}_{A} - K^{30}_{AB})&=&\psi_{13}\psi_{22}-\psi_{12}\psi_{23}.
\end{eqnarray}
The six linear combinations in Eq. (\ref{trex4}) are the determinants of the six $2\times 2$ submatrices that involve only elements of the second and third row of $\Psi_{AB}$.

Next, we consider the six linear combinations
\begin{eqnarray}\label{trex5}
\frac{1}{4} (K^{11}_{AB} - i K^{12}_{AB} + i K^{21}_{AB} + K^{22}_{AB})&=&\psi_{12}\psi_{30}-\psi_{10}\psi_{32},\nonumber\\
\frac{1}{4} (K^{11}_{AB} + i K^{12}_{AB} + i K^{21}_{AB} - K^{22}_{AB})&=&\psi_{11}\psi_{33}-\psi_{13}\psi_{31},\nonumber\\ 
\frac{1}{4} (K^{1}_{A} + i K^{2}_{A} +  K^{10}_{AB} +i K^{20}_{AB})&=&\psi_{11}\psi_{30}-\psi_{10}\psi_{31},\nonumber\\
\frac{1}{4} (K^{1}_{A} + i K^{2}_{A} -  K^{10}_{AB} -i K^{20}_{AB})&=&\psi_{13}\psi_{32}-\psi_{12}\psi_{33},\nonumber\\
\frac{1}{4} (L^{1}_{A} + i L^{2}_{A} +  K^{13}_{AB} +i K^{23}_{AB})&=&\psi_{11}\psi_{32}-\psi_{12}\psi_{31},\nonumber\\
\frac{1}{4} (L^{1}_{A} + i L^{2}_{A} -  K^{13}_{AB} -i K^{23}_{AB})&=&\psi_{13}\psi_{30}-\psi_{10}\psi_{33}.
\end{eqnarray}
The six linear combinations in Eq. (\ref{trex5}) are the determinants of the six $2\times 2$ submatrices that involve only elements of the second and fourth row of $\Psi_{AB}$.

Finally, we consider the six linear combinations
\begin{eqnarray}\label{trex6}
\frac{1}{4} (K^{1}_{B} - i K^{2}_{B} -  K^{01}_{AB} +i K^{02}_{AB})&=&\psi_{20}\psi_{32}-\psi_{22}\psi_{30},\nonumber\\
\frac{1}{4} (K^{1}_{B} + i K^{2}_{B} -  K^{01}_{AB} -i K^{02}_{AB})&=&\psi_{23}\psi_{31}-\psi_{21}\psi_{33},\nonumber\\
\frac{1}{4} (I_1 - K^{0}_{A} - K^{0}_{B} + K^{00}_{AB})&=&\psi_{22}\psi_{33}-\psi_{23}\psi_{32},\nonumber\\
\frac{1}{4} (I_1 - K^{0}_{A} + K^{0}_{B} - K^{00}_{AB})&=&\psi_{20}\psi_{31}-\psi_{21}\psi_{30},\nonumber\\
\frac{1}{4} (I_{2A} - L^{0}_{A} + K^{3}_{B} - K^{03}_{AB})&=&\psi_{22}\psi_{31}-\psi_{21}\psi_{32},\nonumber\\
\frac{1}{4} (I_{2A} - L^{0}_{A} - K^{3}_{B} + K^{03}_{AB})&=&\psi_{20}\psi_{33}-\psi_{30}\psi_{23}.
\end{eqnarray}
The six linear combinations in Eq. (\ref{trex6}) are the determinants of the six $2\times 2$ submatrices that involve only elements of the third and fourth row of $\Psi_{AB}$.

Together Eq. (\ref{trex}), Eq. (\ref{trex2}), Eq. (\ref{trex3}), Eq. (\ref{trex4}), Eq. (\ref{trex5}), and Eq. (\ref{trex6}) contain the determinants of all the 36 different $2\times 2$ submatrices of $\Psi_{AB}$.
We see that if $I_1$, $I_2$, $I_{2A}$, $I_{2B}$, $K_A$, $K_B$, $L_A$, $L_B$, and $K_{AB}$ are simultaneously zero all the 36 determinants are zero.

We can conclude that if $I_1$, $I_2$, $I_{2A}$, $I_{2B}$, $K_A$, $K_B$, $L_A$, $L_B$, and $K_{AB}$ are simultaneously zero $\Psi_{AB}$ is at most rank one. Since by definition $\Psi_{AB}\neq 0$ it is rank 1. Thus $\Psi_{AB}(x_A,x_B)$ is a product state and can be written as $\psi(x_A)\otimes \varphi(x_B)^T$ for some spinors $\psi(x_A)$ and $\varphi(x_B)$.

\end{proof}

\section{Discussion and Conclusions}\label{diss}
In this work we have considered the problem of constructing locally Lorentz covariant bitensors that are indicators of spinor entanglement for two Dirac particles held by spacelike separated laboratories. Bitensors of this kind have been previously described in Refs. \cite{spinorent,lorent}. We reviewed some properties of the Dirac equation, the Dirac gamma matrices, as well as the Lorentz group  and its spinor representation. It was then described how to construct Lorentz covariants from skew-symmetric bilinear forms.
We made the physical assumption that the local curvature of spacetime is small enough to be neglected in both laboratories and that each particle can be described as being in its own Minkowski space.  
Furthermore, we assumed that the tensor products of the local one-particle states is a basis for the two-particle states.

Given the physical assumptions we described how to construct locally Lorentz covariant bitensors
for a system of two spacelike separated Dirac particles.
In particular we constructed such bitensors that are identically zero for all product states of the two particles.
We introduced five locally Lorentz covariant bitensors of this kind in addition to four bitensors described in Ref. \cite{spinorent}. The four bitensors in Ref. \cite{spinorent} transform under local Lorentz transformations in the two laboratories as a bi-scalar, a bi-pseudoscalar, a scalar-pseudoscalar and a pseudoscalar-scalar, respectively. The five additional bitensors constructed in this work transform as a vector-scalar, a scalar-vector, a vector-pseudoscalar, a pseudoscalar-vector, and a bi-vector, respectively.
It was shown that this collection of nine locally Lorentz covariant bitensors has the property that the nine bitensors are simultaneously zero if and only if the state of the two Dirac particles is a product state. Thus this collection of bitensors indicates all spinor entangled states of two spacelike separated Dirac particles.

\appendix
\section{The dimension of an operationally motivated Hilbert space}\label{opp}
In the mathematical framework of quantum mechanics experimental propositions about a system being described correspond to subspaces of a Hilbert space. In particular, orthogonal subspaces correspond to mutually exclusive propositions (See e.g. Reference \cite{birkhoff} or Reference \cite{neumann} for a discussion).

Therefore, in an operational description the Hilbert space needs at most as many orthogonal basis vectors as there are experimental propositions that can be made about the system.
In any real world experiment at most finitely many preparations and measurements are ever made.
Therefore a Hilbert space that is operationally constructed can always be chosen as finite dimensional.


\begin{thebibliography}{xx}
\bibitem{epr} A. Einstein, B. Podolsky, and N. Rosen,  Phys. Rev. {\bf 47}, 777 (1935).
\bibitem{bell}J. S. Bell, Physics {\bf 1}, 195 (1964).
\bibitem{chsh}J. F. Clauser, M. A. Horne, A. Shimony, and R. A. Holt, Phys. Rev. Lett. {\bf 23}, 880 (1969).
\bibitem{bell2}J. S. Bell, Epistemol. Lett. {\bf 9}, 11 (1976).
\bibitem{schrodinger1}E. Schr\"odinger, Proc. Camb. Phil. Soc. {\bf 31}, 555 (1935).
\bibitem{schrodinger3} E. Schr\"odinger, Naturwissenschaften {\bf 23}, 807 (1935).
\bibitem{schrodinger2}E. Schr\"odinger, Proc. Camb. Phil. Soc. {\bf 32}, 446 (1936).












\bibitem{wise}H. M. Wiseman, S. J. Jones, and A. C. Doherty,
Phys. Rev. Lett. {\bf 98}, 140402 (2007).
\bibitem{wiesner}C. H. Bennett and S. J. Wiesner,
Phys. Rev. Lett. {\bf 69}, 2881 (1992).
\bibitem{bennett}C. H. Bennett, G. Brassard, C. Cr{\'e}peau, R. Jozsa, A. Peres, and
W. K. Wootters, Phys. Rev. Lett. {\bf 70}, 1895 (1993).

\bibitem{ghz} D. M. Greenberger, M. Horne, and A. Zeilinger, in {\it Bell's Theorem,  Quantum  Theory,  and  Conceptions  of  the  Universe}, edited  by  M. Kafatos (Kluwer, Dordrecht, 1989).

\bibitem{ekert}A. Ekert and P. L. Knight, Am. J. Phys. {\bf 63}, 415 (1995).
\bibitem{wootters}S. Hill and W. K. Wootters, Phys. Rev. Lett. {\bf 78}, 5022 (1997).
\bibitem{grassl}M. Grassl, M. R\"otteler, and T. Beth, Phys. Rev. A {\bf 58}, 1833 (1998).

\bibitem{wootters2}W. K. Wootters, Phys. Rev. Lett. {\bf 80}, 2245 (1998).
\bibitem{popescu}N. Linden and S. Popescu, Fortsch. Phys. {\bf 46}, 567 (1998).

\bibitem{pop2}H. A. Carteret, N. Linden, S. Popescu and A. Sudbery, Found. Phys. {\bf 29}, 527 (1999).




\bibitem{lind2}N. Linden, S. Popescu, and A. Sudbery, Phys. Rev. Lett. {\bf 83}, 243 (1999).
\bibitem{kempe}J. Kempe, Phys. Rev. A {\bf 60}, 910 (1999).
\bibitem{toni}A. Ac{\'i}n, A. Andrianov, L. Costa, E. Jan{\'e}, J. I. Latorre, and R. Tarrach,
Phys. Rev. Lett. {\bf 85}, 1560 (2000).
\bibitem{coffman}V. Coffman, J. Kundu, and W. K. Wootters,
Phys. Rev. A {\bf 61}, 052306 (2000).
\bibitem{dur}W. D\" ur, G. Vidal, and J. I. Cirac, Phys. Rev. A {\bf 62}, 062314 (2000).
\bibitem{car2} H. A. Carteret and A. Sudbery, J. Phys. A: Math. Gen. {\bf 33}, 4981 (2000).
\bibitem{higuchi}H. A. Carteret, A. Higuchi, and A. Sudbery, J. Math. Phys. {\bf 41}, 7932 (2000).
\bibitem{sud} A. Sudbery, J. Phys. A: Math. Gen. {\bf 34}, 643 (2001).



\bibitem{tarrach} A. Ac{\'i}n, A. Andrianov, E. Jan{\'e} and R. Tarrach, J. Phys. A: Math. Gen. {\bf 34}, 6725 (2001).



\bibitem{verstraete2}F. Verstraete, J. Dehaene, B. De Moor, and H. Verschelde,
Phys. Rev. A {\bf 65}, 052112 (2002).
\bibitem{luque}J.-G. Luque and J.-Y. Thibon, Phys. Rev. A {\bf 67}, 042303 (2003).








\bibitem{ruse}H. S. Ruse, Q. J. Math. {\bf os-2}, 190 (1931).
\bibitem{ruse2}H. S. Ruse, Proc. Lond. Math. Soc. {\bf 32}, 87 (1931).
\bibitem{synge2}J. L. Synge, Proc. Lond. Math. Soc. {\bf 32}, 241 (1931).
\bibitem{synge}J. L. Synge, {\it Relativity: The General Theory} (North Holland, Amsterdam, 1960), Ch. II \S 1.

\bibitem{dewitt}B. S. DeWitt, R. W. Brehme, Ann. Phys. (N. Y.) {\bf 9}, 220 (1960).





\bibitem{dirac2}P. A. M. Dirac, Proc. Royal Soc. A {\bf 117}, 610 (1928).
\bibitem{dirac} P. A. M. Dirac,
{\it Principles of Quantum Mechanics, Fourth edition} (Oxord University Press,
London, 1958).

\bibitem{schwartz} M. D. Schwartz,
{\it Quantum Field Theory and the Standard Model} (Cambridge University Press,
Cambridge, 2014).

\bibitem{yukawa}H. Yukawa, Proc. Phys. Math. Soc. Japan {\bf 17}, 48 (1935).


\bibitem{czachor}M. Czachor, Phys. Rev. A {\bf 55}, 72 (1997).
\bibitem{alsing}P. M. Alsing and G. J. Milburn, Quantum Inf. Comput. {\bf 2}, 487 (2002).
\bibitem{terno} A. Peres, P. F. Scudo, and D. R. Terno,
Phys. Rev. Lett. {\bf 88}, 230402 (2002).

\bibitem{adami}R. M. Gingrich and C. Adami,
Phys. Rev. Lett. {\bf 89}, 270402 (2002).

\bibitem{pachos}J. Pachos and E. Solano, Quantum Inf. Comput. {\bf 3}, 115 (2003).
\bibitem{ahn}D. Ahn, H.-j. Lee, Y. H. Moon, and S. W. Hwang,
Phys. Rev. A {\bf 67}, 012103 (2003).
\bibitem{terno2}D. R. Terno,
Phys. Rev. A {\bf 67}, 014102 (2003).
\bibitem{tera}H. Terashima and M. Ueda, Quantum Inf. Comput. {\bf 3}, 224 (2003).

\bibitem{tera2}H. Terashima and M. Ueda, Int. J. Quantum Inform. {\bf 1},  93 (2003).
\bibitem{mano}E. B. Manoukian and N. Yongram, Eur. Phys. J. D {\bf 31}, 137 (2004).
\bibitem{won}W. T. Kim and E. J. Son,
Phys. Rev. A {\bf  71}, 014102 (2005).
\bibitem{caban3}P. Caban and J. Rembieli{\'n}ski, Phys. Rev. A {\bf 72}, 012103 (2005).
\bibitem{leon}L. Lamata, J. Le{\'o}n, and E. Solano,
Phys. Rev. A {\bf 73}, 012335 (2006).

\bibitem{caban}P. Caban and J. Rembieli{\'n}ski, Phys. Rev. A {\bf 74}, 042103 (2006).
\bibitem{tessier}P. M. Alsing, I. Fuentes-Schuller, R. B. Mann, and T. E. Tessier,
Phys. Rev. A {\bf 74}, 032326 (2006).
\bibitem{geng}H-J. Wang and W. T. Geng,  J. Phys. A: Math. Theor. {\bf 40}, 11617 (2007).
\bibitem{delgado}A. Bermudez and M. A. Martin-Delgado, J. Phys. A: Math. Theor. {\bf 41}, 485302 (2008).
\bibitem{moradi}S. Moradi, Jetp Lett. {\bf 89}, 50 (2009).
\bibitem{caban2}P. Caban, J. Rembieli{\'n}ski, and M. W\l odarczyk, Phys. Rev. A {\bf 79}, 014102 (2009).



\bibitem{spinorent}M. Johansson, Phys. Rev. A {\bf 105}, 032402 (2022).
\bibitem{multispinor}M. Johansson, Ann. Phys. (N. Y.) {\bf 457}, 169410 (2023).
\bibitem{lorent}M. Johansson, arXiv:2308.00896 (2023).

\bibitem{pauli}W. Pauli, Ann. de l'Inst. Henri Poincar\'{e} {\bf 6}, 109 (1936).

\bibitem{navascues}M. Navascu{\'e}s, S. Pironio, and A. Ac{\'i}n,
Phys. Rev. Lett. {\bf 98}, 010401 (2007).
\bibitem{tsirelson}B. S. Tsirelson, {\it Bell inequalities and operator algebras:
http://www.imaph.tu-bs.de/qi/problems/33.html}, (2006).
\bibitem{werner}V. B. Scholz and R. F. Werner, arXiv:0812.4305 (2008).
\bibitem{wald}R. M. Wald, {\it General Relativity} (The University of Chicago Press,
Chicago, 1984).


\bibitem{zuber} C. Itzykson and J-B. Zuber,
{\it Quantum Field Theory} (Dover,
New York, 2006), Ch. 2-1-3.

\bibitem{bjorken} J. D. Bjorken and S. D. Drell,
{\it Relativistic Quantum Mechanics} (McGraw-Hill,
New York, 1964).



\bibitem{birkhoff}G. Birkhoff and J. Von Neumann, Ann. Math. {\bf 37}, 823 (1936).

\bibitem{neumann}J. Von Neumann, {\it Mathematical Foundations of Quantum Mechanics} (Princeton University Press,
Princeton, 1955), Ch. III.5.

\end{thebibliography}
\end{document}